\documentclass[letter,11pt]{article}

\usepackage[utf8]{inputenc}
\usepackage[T1]{fontenc}

\usepackage[a4paper,margin=1in]{geometry}

\usepackage[usenames,dvipsnames]{xcolor}

    \definecolor{DarkRed}{rgb}{0.368,0.097,0.078}
\definecolor{DarkBlue}{rgb}{0.2,0.2,0.6}

\usepackage[colorlinks = true,
    linkcolor = DarkBlue,
    anchorcolor = DarkBlue,
    citecolor = DarkBlue,
    filecolor = DarkBlue,
    urlcolor = DarkBlue,
    pagebackref]{hyperref}

\usepackage{makecell}
\usepackage{multirow}

\usepackage{amsthm} 
\usepackage{mathtools,thmtools}
\usepackage{amsfonts,amsmath,amssymb}
\usepackage[capitalise]{cleveref}
\usepackage{times}
\usepackage{algorithm}
\usepackage{algpseudocode}
\usepackage{thm-restate}
\usepackage{tcolorbox}
\usepackage{lipsum}
\usepackage{enumitem}
\usepackage{bm}
\usepackage{xspace}
\usepackage{nicefrac}
\usepackage{dsfont}
\usepackage{float}
\usepackage{mdframed}

\usepackage{bbm}
\usepackage{mathtools,thmtools}

\declaretheoremstyle[
	    spaceabove=\topsep, 
	    spacebelow=\topsep, 
	    headfont=\normalfont\bfseries,
	    bodyfont=\normalfont\itshape,
	    notefont=\normalfont\bfseries,
	    notebraces={(}{)},
	    postheadspace=0.5em, 
	    headpunct={},
	    postfoothook=\noindent\ignorespaces
    ]{theorem}
\declaretheorem[style=theorem,numberwithin=section]{theorem}

\declaretheoremstyle[
	    spaceabove=\topsep, 
	    spacebelow=\topsep, 
	    headfont=\normalfont\bfseries,
	    bodyfont=\normalfont,
	    notefont=\normalfont\bfseries,
	    notebraces={(}{)},
	    postheadspace=0.5em, 
	    headpunct={},
	    postfoothook=\noindent\ignorespaces
    ]{definition}

\declaretheoremstyle[
        spaceabove=\topsep, 
        spacebelow=\topsep, 
        headfont=\normalfont\bfseries,
        bodyfont=\normalfont,
        notefont=\normalfont\bfseries,
        notebraces={}{},
        postheadspace=0.5em, 
        qed=$\blacksquare$, 
        headpunct={},
        postfoothook=\noindent\ignorespaces
    ]{proofstyle}
\declaretheorem[style=proofstyle,numbered=no,name=Proof]{proof}

\declaretheorem[style=theorem,sibling=theorem,name=Corollary]{corollary}

\declaretheorem[style=theorem,sibling=theorem,name=Conjecture]{conjecture}

\declaretheorem[style=theorem,numbered=no,name=Theorem]{theorem*}
\declaretheorem[style=theorem,numbered=no,name=Lemma]{lemma*}
\declaretheorem[style=theorem,numbered=no,name=Corollary]{corollary*}
\declaretheorem[style=theorem,numbered=no,name=Proposition]{proposition*}
\declaretheorem[style=theorem,numbered=no,name=Claim]{claim*}
\declaretheorem[style=theorem,numbered=no,name=Fact]{fact*}
\declaretheorem[style=theorem,numbered=no,name=Observation]{observation*}
\declaretheorem[style=theorem,numbered=no,name=Conjecture]{conjecture*}

\declaretheorem[style=definition,sibling=theorem,name=Definition]{definition}

\declaretheorem[style=definition,numbered=no,name=Definition]{definition*}
\declaretheorem[style=definition,numbered=no,name=Remark]{remark*}
\declaretheorem[style=definition,numbered=no,name=Example]{example*}
\declaretheorem[style=definition,numbered=no,name=Question]{question*}

\newcommand{\lr}[1]{\mathopen{}\left(#1\right)}

\newcommand{\norm}[1]{\|#1\|}

\usepackage[normalem]{ulem}
\newcommand{\idan}[1]{{\color{purple} Idan: {#1}}}
\newcommand{\XG}[1]{{\color{teal} [Xing: #1]}}

\newcommand{\onenorm}[1]{\norm{#1}_1}
\newcommand{\twonorm}[1]{\norm{#1}_2}
\newcommand{\inftynorm}[1]{\norm{#1}_\infty}
\newcommand{\fnorm}[1]{{\norm{#1}}_F}
\newcommand{\OPT}{\mathrm{OPT}}
\newcommand{\R}{\mathbb{R}}

\newcommand{\numSAT}{\# \mathrm{SAT}}
\newcommand{\numLIN}{\# \mathrm{LIN}}

\providecommand{\expect}[2]{\ensuremath{\ifthenelse{\equal{#1}{}}{\mathbb{E}}{\mathbb{E}_{#1}}\!\left[#2\right]}\xspace}
\providecommand{\expectUnder}[2]{\ensuremath{\ifthenelse{\equal{#1}{}}{\mathbb{E}}{\underset{#1}{\mathbb{E}}}\!\left[#2\right]}\xspace}
\providecommand{\prob}[2]{\ensuremath{\ifthenelse{\equal{#1}{}}{\Pr}{\mathbb{P}_{#1}}\!\left[#2\right]}\xspace}
\algrenewcommand\algorithmicrequire{\textbf{Input:}}
\algrenewcommand\algorithmicensure{\textbf{Output:}}

\newcommand{\ksat}{$k\text{-SAT }$}
\newcommand{\eksat}{E$k\text{-SAT }$}
\newcommand{\threesat}{$3\text{-SAT }$}

\newcommand{\maxksat}{$\text{MAX-}k\text{-SAT }$}
\newcommand{\maxthreesat}{$\text{MAX-}3\text{-SAT }$}
\newcommand{\maxtwosat}{$\text{MAX-}2\text{-SAT }$}
\newcommand{\maxeksat}{$\text{MAX-E}k\text{-SAT }$}

\renewcommand{\idan}[1]{}
\renewcommand{\XG}[1]{}

\title{Learning-Augmented Algorithms for Boolean Satisfiability}

\author{
Idan Attias\thanks{
Institute for Data, Econometrics, Algorithms, and Learning (IDEAL), hosted by the
University of Illinois at Chicago and Toyota Technological Institute at Chicago; \texttt{idanattias88@gmail.com}.}
\and Xing Gao\thanks{University of Illinois at Chicago; \texttt{xgao53@uic.edu}}
\and Lev Reyzin\thanks{University of Illinois at Chicago; \texttt{lreyzin@uic.edu}}
}

\date{}
\begin{document}
\maketitle
    
\begin{abstract}
Learning-augmented algorithms are a prominent recent development in beyond worst-case analysis. In this framework, a problem instance is provided with a prediction (``advice'') from a machine-learning oracle, which provides partial information about an optimal solution, and the goal is to design algorithms that leverage this advice to improve worst-case performance. We study the classic Boolean satisfiability (SAT) decision and optimization problems within this framework using two forms of advice. ``Subset advice" provides a random $\epsilon$ fraction of the variables from an optimal assignment, whereas ``label advice" provides noisy predictions for all variables in an optimal assignment.

For the decision problem \ksat, by using the subset advice we accelerate the exponential running time of the PPSZ family of algorithms due to Paturi, Pudlak, Saks and Zane \cite{paturi2005improved}, which currently represent the state of the art in the worst case. We accelerate the running time by a multiplicative factor of $2^{-c}$ in the base of the exponent, where $c$ is a function of $\epsilon$ and $k$.
In particular, for \threesat the base constant of the exponent becomes $2^{\left(\frac{\epsilon}{1-\epsilon}+2\ln{(2-2\epsilon)}-1+o(1)\right)}$ for $\epsilon<\frac{1}{2}$ comparing to $2^{2\ln{2} - 1 +o(1)}$ without advice, and when $\epsilon\geq 1/2$, the running time becomes sub-exponential.
For the optimization problem, we show how to incorporate subset advice in a black-box fashion with any $\alpha$-approximation algorithm, improving the approximation ratio to $\alpha + (1 - \alpha)\epsilon$.
Specifically, we achieve approximations of $0.94 + \Omega(\epsilon)$ for \maxtwosat, $7/8 + \Omega(\epsilon)$ for \maxthreesat, and $0.79 + \Omega(\epsilon)$ for MAX-SAT. Moreover, for label advice, we obtain near-optimal approximation for instances with large average degree, thereby generalizing recent results on MAX-CUT and MAX-$2$-LIN \cite{ghoshal2024constraint}.
\end{abstract}

\section{Introduction}

The Boolean Satisfiability problem (SAT) is a cornerstone of computational complexity theory and algorithmic research. Given a Boolean formula over variables that can take values true or false, the task is to decide whether there exists an assignment of these variables that makes the formula evaluate to true. SAT is most commonly studied in its conjunctive normal form (CNF), where the formula is represented as a conjunction of clauses, each clause being a disjunction of literals (a variable or its negation). The restricted case where each clause contains at most $k$ literals is known as \ksat. SAT was the first problem proven to be NP-complete, via the Cook–Levin theorem \cite{cook1971complexity,Levin73}, with $3\text{-SAT}$ often serving as its canonical example \cite{karp1972reducibility}. This foundational result implies that any problem in NP can be efficiently reduced to SAT, making it a central object in the study of computational intractability, reductions, and practical solving techniques.

Since polynomial-time algorithms that solve all instances of SAT optimally are unlikely to exist unless $\text{P}\neq\text{NP}$, alternative approaches are necessary. One approach is to abandon the requirement of polynomial runtime and seek exponential-time algorithms that, while still exponential in the worst case, are faster than exhaustive search. For $k \geq 3$, let $c_k \ge 1$ be a constant such that \ksat can be solved in time $O^*(c_k)^n$, where $n$ is the number of variables in the given \ksat instance and $O^*(\cdot)^n$ hides polynomial factors. A well-known conjecture proposed by \cite{impagliazzo2001complexity}, called the Exponential Time Hypothesis (ETH), posits that \threesat cannot be solved in sub-exponential time, i.e., $c_3 > 1$. A positive answer to this conjecture would imply that $\text{P} \neq \text{NP}$. A stronger conjecture proposed by \cite{calabro2009complexity}, known as the Strong Exponential Time Hypothesis (SETH), claims that
$\lim_{k \to \infty} c_k = 2$.

Another canonical variant of the Boolean satisfiability problem is its optimization counterpart, MAX-SAT, where the objective is to determine the maximum number of clauses that can be satisfied by any assignment. Since MAX-SAT generalizes the decision problem SAT, it is also computationally intractable unless $\text{P} = \text{NP}$. One approach to tackling this problem is to design polynomial-time algorithms that find approximate solutions for all instances. These are known as approximation algorithms. However, it is computationally hard to compute an approximate solution that satisfies a number of clauses arbitrarily close to the optimal. More precisely, MAX-SAT is APX-complete, indicating that it does not admit a polynomial-time approximation scheme (PTAS) unless $\text{P} = \text{NP}$ \cite{feige1995approximating,arora1998probabilistic,arora1998proof,haastad2001some}. Therefore, there has been extensive study to find the best possible approximation factors for this problem, particularly for \maxksat, the restricted case of MAX-SAT where each clause contains at most $k$ literals.

There has been extensive research on SAT beyond worst-case performance, including random (average-case) and semirandom models (\cite{selman1996generating}, \cite[Section 9]{roughgarden2021beyond}), smoothed analysis \cite{feige2007refuting}, stability analysis~\cite[Section 2]{KunR14}, parameterized complexity \cite[Section 2]{roughgarden2021beyond}, and SAT-solver heuristics such as DPLL and CDCL (\cite[Section 25]{roughgarden2021beyond}).
We focus on the emerging paradigm of \emph{learning-augmented algorithms} \cite{mitzenmacher2022algorithms}, also known as “algorithms with predictions”. In this paradigm, a machine learning method provides a prediction (or “advice”) about the input or the optimal solution, and the algorithm uses this prediction to improve its performance, with guarantees that depend on the accuracy of the prediction. This approach has been applied to numerous algorithmic tasks, particularly NP problems such as MAX-CUT, Max-$k$-Lin, Independent Set and Clustering, among others \cite{ghoshal2024constraint, cohen2024learning, braverman2024learning, bampis2024parsimonious,dong2024learning,bampis2025polynomial,ergunlearning,gamlath2022approximate,nguyen2022improved}. This paper focuses on studying the canonical NP-complete problem of SAT.

\subsection{Problem Formulation}

Let $\phi(x)=C_1 \land C_2 \land \dots \land C_m$ be a Boolean formula in conjunctive normal form (CNF) over variables $x = x_1, x_2, \dots, x_n$, consisting of $m$ clauses $C_1, C_2, \dots, C_m$.%

\begin{definition}[SAT]
    The Satisfiability Problem asks whether there exists a truth assignment $\sigma : \{x_1, \dots, x_n\} \to \{0,1\}$ to a CNF Boolean $\phi$ formula such that $\phi(x)|_{\sigma} = 1$, i.e., all clauses $C_i$ are satisfied under $\sigma$.
    When each clause in $\phi$ contains at most $k$ literals, the problem is referred to as \ksat. If every clause contains exactly $k$ literals, it is known as \eksat.
\end{definition}

\begin{definition}[MAX-SAT]
     Given a CNF Boolean formula $\phi$, the MAX-SAT problem seeks a truth assignment $\sigma : \{x_1, \dots, x_n\} \to \{0,1\}$ that maximizes the number of clauses $C_i$ for which $C_i|{_\sigma} = 1$, i.e., the number of clauses satisfied by $\sigma$.
    When each clause in $\phi$ contains at most $k$ literals, the problem is referred to as \maxksat. If every clause contains exactly $k$ literals, it is known as \maxeksat.
\end{definition}

Let $x^* = \{x^*_1, \ldots, x^*_n\}$ be a fixed optimal solution to either SAT or MAX-SAT. We consider two models of advice: one provides full certainty on a subset of the values in $x^*$, and the other offers noisy information about all values in $x^*$.

\begin{definition}[Subset Advice]
The subset advice consists of a random subset of indices $S \subset \{1, \ldots, n\}$ along with the ground-truth assignment $x^*$ restricted to these indices $\{x_i^*\}_{i \in S}$, where each $i$ is included in $S$ independently with probability $\epsilon$ for all $i = 1, \ldots, n$.
\end{definition}

\begin{definition}[Label Advice]
The label advice is an assignment $\tilde{x}\in \{0, 1\}^n$ that contains noise relative to a ground-truth optimal assignment $x^*$. For each $i = 1, \ldots, n$, independently, we have
\[
\tilde{x_i} = \begin{cases}
    x_i^* & \text{ with probability } \frac{1+\epsilon}{2},\\
    1-x_i^* & \text{ with probability } \frac{1-\epsilon}{2}.
\end{cases}
\]
\end{definition}

It is important to note that the randomness in the label advice is sampled once and fixed (i.e., the oracle returns the same answer if queried multiple times), which is standard in the literature on learning-augmented algorithms. Otherwise, one could trivially boost the algorithm’s probability of success. See the discussion of persistent vs. non-persistent noise in \cite{braverman2024learning}.

Moreover, as noted by \cite{ghoshal2024constraint}, the subset advice model is stronger than the label advice model, since label advice $\tilde{x}$ can be simulated given subset advice.

\subsection{Our Contributions}
Under the learning-augmented framework, we study the SAT and MAX-SAT problems with advice. Here, we highlight our contributions and provide a road map of the paper.

In Section \ref{sec: decision problem}, we study the decision problem of \ksat with advice and improve the running time of state-of-the-art algorithms. Specifically, we incorporate subset advice into a family of algorithms known as the PPZ (Section \ref{sec: subset ppz}) and PPSZ (Section \ref{sec: subset ppsz}). As both algorithms run in exponential time in the worst case, we present our results as improvements on the base constant of the exponent, summarized in Table \ref{table:decision}.
For PPZ, we improve the running time from $2^{(1-\frac{1}{k})}$ to $2^{(1-\frac{1}{k}\frac{1-\epsilon^{k}}{1-\epsilon})(1-\epsilon)}$, and for PPSZ we improve the base of the exponent by a multiplicative factor of $2^{-c(\epsilon,k)}$, where $c(\epsilon,k)>0$. In particular, for \threesat the base constant of the exponent becomes $2^{\left(\frac{\epsilon}{1-\epsilon}+2\ln{(2-2\epsilon)}-1+o(1)\right)}$ for $\epsilon<\frac{1}{2}$ comparing to $2^{2\ln{2} - 1 +o(1)}$ without advice, and when $\epsilon\geq 1/2$, the running time becomes sub-exponential.
We also provide a hardness result for \threesat with subset advice under ETH (Section \ref{sec: subset hardness decision}).

In Section \ref{sec: optimization problem}, we study the optimization problem of MAX-SAT and its variants \maxksat, aiming to improve the approximation factors of polynomial-time algorithms. With subset advice (Section \ref{sec: subset maxsat}), we show that any $\alpha$-approximation algorithm for a variant of MAX-SAT can be turned into an approximation of $\alpha+(1-\alpha)\epsilon$ by incorporating the subset advice in a black-box fashion. To complement this result, we establish hardness of approximation for \maxthreesat with subset advice, assuming Gap-ETH (Section \ref{sec: subset hardness optimization}). With label advice, we focus specifically on the \maxtwosat problem, extending the quadratic programming approach for MAX-$2$-LIN from \cite{ghoshal2024constraint} (Section \ref{sec: label max2sat}) to obtain a near-optimal approximation for instances whose average degree exceeds a threshold depending only on the amount of advice.

\paragraph{Comparison to prior work}
Under the same subset and label advice models, the recent works of \cite{ghoshal2024constraint,cohen2024learning} both study closely related problems: MAX-CUT, MAX-$k$-LIN, and MAX-$2$-CSP. In particular, \cite{ghoshal2024constraint} studies MAX-$k$-LIN for $k = 2, 3, 4$, which includes MAX-CUT as a special case. For their positive results, they design algorithms in the weaker label advice model and give near-optimal solutions to MAX-$2$-LIN, under the assumption that the instance has large average degree. For their negative results, they show conditional hardness for MAX-$3$-LIN and MAX-$4$-LIN in the stronger subset advice model. The work of \cite{cohen2024learning} provides positive results in both advice models. With label advice, they achieve an $\alpha_{GW} + \Omega(\epsilon^4)$ approximation for MAX-CUT and MAX-$2$-CSP based on a notion of wide'' and narrow'' graphs and their respective properties. With subset advice, they achieve $\alpha_{GW} + \Omega(\epsilon^2)$ and $\alpha_{RT} + \Omega(\epsilon)$ approximations for MAX-CUT\footnote{$\alpha_{GW} = 0.878$ \cite{goemans1994879}, $\alpha_{RT} = 0.858$ \cite{raghavendra2012approximating}.}. Both of these works study the optimization problem, but not the decision problem.

\subsection{Related Work}

\paragraph{Learning-Augmented Algorithms} 
The idea of using additional information, such as a prediction about the future or a suggestion about the solution, to improve an algorithm’s performance originated in the field of online algorithms \cite{devanur2009adwords,vee2010optimal}.
The formal framework was introduced by Lykouris and Vassilvitskii \cite{lykouris2021competitive}, who defined the key notions of consistency and robustness for evaluating algorithm performance. Consistency refers to improved performance when the predictions are accurate, while robustness ensures that the algorithm performs comparably to a standard, prediction-free algorithm even when the predictions are unreliable.
Numerous algorithmic problems have been studied in this framework, including caching, paging, the ski rental problem, online bipartite matching, scheduling, load balancing, and online facility location \cite{purohit2018improving,mitzenmacher2019scheduling,lattanzi2020online,rohatgi2020near,wei2020optimal,lykouris2021competitive,antoniadis2023paging}.
For broader context, see the survey by Mitzenmacher and Vassilvitskii \cite{mitzenmacher2022algorithms} and the online database of papers in the field \cite{ALP2023}.

\paragraph{SAT} Improvements over exhaustive search in the worst case were achieved by a family of deterministic algorithms based on branching \cite{monien1985solving,schiermeyer1993solving,rodovsek1996new,kullmann1999new}. Another family of algorithms is based on local search, initiated by the randomized algorithm of Sch\"oning
 \cite{schoning1999probabilistic}, and later improved (and in some cases derandomized) by subsequent works \cite{hofmeister2002probabilistic,baumer2003improving,iwama2004improved,moser2011full,liu2018chain}. A third family of randomized algorithms is based on random restrictions, initiated by Paturi, Pudlak and Zane \cite{paturi1997satisfiability} and Paturi, Pudl\'ak, Saks, and Zane \cite{paturi2005improved} algorithms. Hertli \cite{hertli20143,hertli2014breaking} improved the analysis of PPSZ, which was later simplified by Scheder and Steinberger \cite{scheder2017ppsz} and slightly improved by Qin and Huang \cite{qin2020improvement}. A variant of PPSZ, named biased PPSZ, was introduced by Hansen, Kaplan, Zamir, and Zwick \cite{hansen2019faster}, and an improved analysis of PPSZ by Scheder \cite{scheder2024ppsz} currently represents the state of the art, with $1.307^n$ for \threesat.

 \paragraph{MAX-SAT} Goemans and Williamson~\cite{goemans1994879} achieved a significant breakthrough by introducing a semidefinite programming (SDP) relaxation combined with randomized rounding, resulting in a $0.878$-approximation algorithm for \maxtwosat. Building upon this, Feige and Goemans~\cite{feige1995approximating} improved the approximation ratio to $0.931$ by applying a rotation method prior to randomized rounding. Zwick~\cite{zwick2000analyzing} further refined the analysis of these algorithms. Subsequently, Matsui and Matuura~\cite{matuura0} improved the approximation ratio to $0.935$ by employing a skewed distribution during the rounding phase. Finally, Lewin, Livnat, and Zwick~\cite{lewin2002improved} combined these techniques (SDP relaxation, rotation, and skewed distribution rounding) to develop an algorithm achieving a $0.940$-approximation for \maxtwosat.
Under the Unique Games Conjecture (UGC), this $0.940$-approximation is proven to be optimal, indicating that no polynomial-time algorithm can achieve a better approximation ratio for \maxtwosat unless the UGC is false~\cite{brakensiek2024tight,austrin2007balanced}. Assuming only that $\text{P} \neq \text{NP}$, the best achievable approximation ratio is approximately $0.954$, as established by H$\mathring{\text{a}}$stad~\cite{haastad2001some}.

A semidefinite programming relaxation technique was also applied to \maxthreesat, where Karloff and Zwick~\cite{karloff19977} achieved a $7/8$-approximation algorithm\footnote{This guarantee holds for satisfiable instances, and there is strong evidence suggesting the algorithm achieves the same approximation ratio for unsatisfiable instances as well.}. For the case where each clause contains exactly three literals (\maxeksat), Johnson~\cite{johnson1973approximation} showed that a simple random assignment achieves a $7/8$-approximation as well. H$\mathring{\text{a}}$stad~\cite{haastad2001some} later proved that this is optimal, showing that no polynomial-time algorithm can achieve a better approximation ratio unless $\text{P} = \text{NP}$. For the general MAX-SAT problem, the current best-known approximation factor is $0.796$, obtained by Avidor, Berkovitch, and Zwick~\cite{avidor2005improved}.

\section{Improvements on the Running Time of the Decision Problem}\label{sec: decision problem}
In this section, we show how subset advice can improve the running time for \ksat using algorithms in the family of ``random restriction algorithms'' \cite{paturi1997satisfiability,paturi2005improved,hansen2019faster,scheder2017ppsz}, initiated by the influential PPZ algorithm by Paturi, Pudl\'ak, and Saks. 
We start with the relatively straightforward PPZ algorithm \cite{paturi1997satisfiability}, and then dive into the more involved PPSZ algorithm \cite{paturi2005improved,hertli20143}. A summary of existing results and our improvements can be found in Table \ref{table:decision}.

We consider the Unique-\ksat problem (i.e., instances with at most one satisfying
assignment) for simplicity. According to \cite{hertli20143,scheder2017ppsz,scheder2024ppsz}, Unique-\ksat bounds can be lifted to general \ksat with multiple satisfying assignments. In particular, the generalization from \cite{hertli20143} use the results on Unique-\ksat in a black-box fashion and show that the same bounds hold for \ksat, therefore we analyze the improvements on Unique-\ksat, and our results apply to \ksat as well. (Cf. Theorems 1, 2, 7 \cite{hertli20143} and Theorem 11 \cite{scheder2017ppsz}.) 

\begin{table}[H]
\begin{center}
\setlength{\arrayrulewidth}{0.2mm} 
\setlength{\tabcolsep}{2.5pt} 
\def\arraystretch{1.5} 
\begin{tabular}{|c|c|c|}
\hline
\multicolumn{3}{c}{\textsc{Exact exponential-time algorithms for $k$-SAT}} \\
\hline
\hline
& Original & Subset Advice \\
\hline
\rule{0pt}{27pt} PPZ 
&\makecell{$c_k=2^{1-1/k}$\\$c_3\approx 1.587$ \\ \cite{paturi1997satisfiability}} 
&\makecell{$c_k'\le c_k^{(1-\epsilon)}$\\ \\ \ [Theorem \ref{thm: subset PPZ}] } 
\\
\hline
\rule{0pt}{33pt} PPSZ 
&\makecell{$c_k=2^{1-R_k+o(1)}$\\$c_3\approx 1.308$ \\ \cite{paturi2005improved}} 
&\makecell{for $\epsilon<\frac{k-2}{k-1}$: $\ c_k' = c_k\cdot 2^{-\epsilon_k}$, where $\epsilon>\epsilon_k>0$\\ for $\epsilon\ge\frac{k-2}{k-1}$: $\ c_k' = 2^{o(1)} \ $, sub-exponential time\\ \ [Theorem \ref{thm: subset PPSZ kSAT}]} 
\\
\hline
\end{tabular}
\end{center}
\caption{All algorithms are of the form $O^*(c_k)^n$, where $O^*(\cdot)^n$ hides polynomial factors. $R_k$ and $\epsilon_k$ are constants depending on $k$. Cf. Equations (\ref{eq: R_k}) and (\ref{eq: eps_k}) for details.}
\label{table:decision}
\end{table}

In Algorithm \ref{algo: subset ppz ppsz}, we incorporate the subset advice into both PPZ and PPSZ: with inputs $S=\varnothing$ (i.e. $\epsilon = 0$) and $D=1$, Algorithm \ref{algo: subset ppz ppsz} recovers the PPZ algorithm. With inputs $S=\varnothing$ (i.e. $\epsilon = 0$) and $D=o\left( \frac{n}{\log n}\right)$, Algorithm \ref{algo: subset ppz ppsz} recovers the PPSZ algorithm. We use the notation $\phi = \phi |_\sigma$ to denote reducing a formula $\phi$ based on some partial assignment $\sigma$, i.e., given $\phi$ and $\sigma$, remove the clauses satisfied by variables in $\sigma$ as well as literals set to false by $\sigma$. 

\begin{algorithm}[H]
\caption{PPZ \cite{paturi1997satisfiability} and PPSZ \cite{hertli20143} with subset advice}\label{algo: subset ppz ppsz}
\textbf{Input:} (i) $k$-CNF formula $\phi$ where $V$ is the set of variables in $\phi$, (ii) random subset advice $S$ where each variable is included with probability $\epsilon$, (iii) implication parameter $D$, (iv) number of iterations $T$.
\\
\textbf{Initialize:} Let $\sigma$ be the empty assignment on $V$.
\begin{enumerate}
    \item For each assignment $b\in \{0,1\}$ for variable $x_i$ given by $S$, let $\sigma(x_i)=b$ and  $\phi = \phi |_{\sigma}$.
    \item Choose a random permutation $\pi$ of the remaining variables.
    \item For each variable $x_i$ in the order of $\pi$:
    \begin{itemize}
        \item Enumerate over all sets of $D$ clauses: if the value of $x_i$ is being forced to be $b\in\{0,1\}$ by some set of these sets (given previously assigned variables), then set $\sigma(x_i)=b$ and  $\phi = \phi |_{\sigma}$.
        \item Else, set $\sigma(x_i)= 0$ or $1$ uniformly at random and $\phi = \phi |_{\sigma}$. 
    \end{itemize}
    \item If $\sigma$ is a satisfying assignment, return $\sigma$. Otherwise, repeat steps $2\text{-}3$ at most $T$ times. If a satisfying assignment was not found, return "unsatisfiable". 
\end{enumerate}
\end{algorithm}

\subsection{PPZ Algorithm for \ksat with Subset Advice}\label{sec: subset ppz}
First, we briefly introduce some analysis of the PPZ algorithm. Although the original analysis of PPZ \cite{paturi1997satisfiability} applies to general \ksat, it is more involved, and we adopt the simplified arguments summarized in \cite{hansen2019faster}, which assume a unique satisfying assignment. 

Under the uniqueness assumption, each variable $x_i$ has a “critical clause” $C_{x_i}$, where the literal associated with $x_i$ is the only one set to true in the clause. The assignments of variables are either “forced” or “guessed”: if a variable $x_i$ appears in the permutation $\pi$ after all other variables in $C_{x_i}$, that is, if during the execution of the algorithm there exists a unit clause involving $x_i$ or $\bar{x_i}$, then $x_i$ is forced, and its literal is set to true. Since the permutation is random, the probability that $x_i$ is forced is at least $1/k$. If the variable is not forced, then it is guessed uniformly at random. The running time improves when more variables are forced.

Denote by $G(\pi)$ the number of guessed variables given the order of $\pi$. If the formula is satisfiable, we can lower bound the success probability of one iteration of the randomized algorithm
\begin{align*}
    \text{Pr[All guessed variables are correct]}
    &=
    \mathbb{E}_\pi\left[2^{-G(\pi)}\right]
    \geq
    2^{-\mathbb{E}_\pi\left[G(\pi)\right]},
\end{align*}
which follows from Jensen's inequality. 
Note that we can analyze the probability of each variable being guessed individually.
We repeat this process for $T = O^*\left(2^{\mathbb{E}_\pi\left[G(\pi)\right]} \right)$ iterations to succeed with high probability, ensuring a Monte Carlo algorithm guarantee. Alternatively, we can have a Las Vegas algorithm that is always correct, but its running time becomes a random variable with expectation $T$.
Thus, without advice, the running time is shown to be $O^*\left(2^{\left(1-\frac{1}{k}\right)n}\right)$. In the following Theorem, we show that the running time can be improved exponentially with subset advice.

\begin{theorem}[PPZ \cite{paturi1997satisfiability} with subset advice]\label{thm: subset PPZ}
Consider the decision problem of \ksat. Given subset advice $S$ where each variable is included independently with probability $\epsilon$, the running time of the PPZ algorithm $O^*(c_k)^n$ can be improved exponentially, in particular, the base constant becomes $c_k'=2^{\left(1-\frac{1}{k}\frac{1-\epsilon^{k}}{1-\epsilon}\right)(1-\epsilon)}$, comparing to $c_k=2^{\left(1-\frac{1}{k}\right)}$ without advice.
\end{theorem}
\begin{proof}
We upper bound $\mathbb{E}_\pi\left[G(\pi)\right]$ by analyzing the probability of a variable being forced in PPZ with subset advice:
\begin{align*}
    \prob{\pi}{x_i \notin S \text{ and is forced}} &\ge \prob{\pi}{x_i \notin S,\text{ and all the other variables in } C_{x_i} \text{ either in $S$ or appear in $\pi$ before } x_i }\\
    &\ge\sum_{j=0}^{k-1}\binom{k-1}{j}\epsilon^{j}(1-\epsilon)^{k-j}\frac{1}{k-j}\quad (\text{where }j \text{ variables in this clause appear in }S)\\
    &=\frac{1}{k}\sum_{j=0}^{k-1}\binom{k}{j}\epsilon^{j}(1-\epsilon)^{k-j}\\
    &=\frac{1}{k}\left(1-\epsilon^{k} \right).
\end{align*}

Therefore $\expect{\pi}{G(\pi)}= (1-\epsilon)n - \frac{1}{k}\left(1-\epsilon^{k} \right)n=n(1-\epsilon)\left(1-\frac{1}{k}\frac{1-\epsilon^{k}}{1-\epsilon}\right)$. The base constant of the running time with subset advice is $c_k'=2^{\left(1-\frac{1}{k}\frac{1-\epsilon^{k}}{1-\epsilon}\right)(1-\epsilon)}$. Comparing to $c_k = 2^{1-\frac{1}{k}}$ of the original PPZ algorithm, we have exponential improvement on the running time depending on the advice.  \end{proof}

\subsection{PPSZ Algorithm for \ksat with Subset Advice}\label{sec: subset ppsz}
The PPSZ algorithm improves upon the PPZ algorithm by introducing a preprocessing step called “$D$-bounded resolution” \cite{paturi2005improved}, which was later relaxed to a concept called “$D$-implication” in an adapted version of PPSZ by \cite{hertli20143}, which we adopt here.
The idea is that a variable can be forced even if not all the variables in the critical clause appeared before it (as analyzed in PPZ). In PPSZ, we force a value for a variable by enumerating over all sets of $D$ clauses to check whether the variable takes the same value in all satisfying assignments consistent with these $D$ clauses. If so, the variable is forced to this value, otherwise, we guess randomly. The probability of a variable being forced increases compared to PPZ, which leads to a better running time.
We take $D = D(n)$ to be a slowly growing function of $n$, e.g., $D = o\bigl(\frac{n}{\log n}\bigr)$, so that the enumeration still runs in reasonable time.

First we summarize the analysis of the PPSZ algorithm \cite{paturi2005improved}. 
We choose the random permutation indirectly.
For each variable, we choose a uniformly random $r \in [0,1]$ representing its ``arrival time”, and we determine the order in the permutation by sorting the arrival times. The reason is that the arrival times are completely independent, in contrast to choosing a random permutation directly, where the positions of two variables are not independent.

By summarizing Lemma 6,7,8 \cite{paturi2005improved}, we get
\begin{align}\label{eq: R_k}
    \prob{\pi}{x_i \text{ is forced}}
    &\ge \int_{0}^{1}R_k(r)dr - \Delta_k^{(d)} = R_k- \Delta_k^{(d)},
\end{align}
where $R_k(r)$ is the smallest nonnegative solution $R$ to $\Big(r+(1-r)R\Big)^{k-1} = R$ and $R_k := \int_{0}^{1}R_k(r)dr$. It is shown that $R_k(r)$ is strictly increasing for $r\in \left[0, \frac{k-2}{k-1}\right]$, and $R_k(r)=1$ for $r\in \left[\frac{k-2}{k-1}, 1\right]$. The asymptotic error of convergence satisfies $0\le \Delta_k^{(d)}\le \frac{3}{(d-1)(k-2)+2}$ where $d$ is the minimum hamming distance between satisfying assignments. Assuming uniqueness, $\Delta_k^{(d)}=o(1)$, and goes to $0$ as $D$ goes to infinity.

It is shown that $R_k = \frac{1}{k-1}\sum_{j=1}^{\infty}\frac{1}{j\left(j+\frac{1}{k-1} \right)}$ for $k\ge3$. Similarly to the PPZ algorithm, we make $T=O^*\left(2^{\mathbb{E}_\pi\left[G(\pi)\right]} \right)$ iterations, where here $\expect{\pi}{G(\pi)} = n \left(1-R_k + \Delta_k^{(d)}\right)$. For $k=3$, we can explicitly evaluate $R_3 = 2- 2\ln{2}$, and $T=O^*\left(2^{(2\ln{2}-1+o(1))} \right)^n$. 

In Theorem \ref{thm: subset PPSZ kSAT} we show that we can improve the running time for \ksat with subset advice, then we state the result explicitly for \threesat in Corollary \ref{thm: subset PPSZ 3SAT}, with its proof deferred to Appendix \ref{apx: ommited proofs decision}.

\begin{restatable}[PPSZ with subset advice for \ksat]{theorem}{PPSZkSAT}\label{thm: subset PPSZ kSAT}
Consider the decision problem of \ksat and suppose without advice the PPSZ algorithm runs in time $O^*(c_k)^n$. Given subset advice $S$ with $0<\epsilon<\frac{k-2}{k-1}$, we improve the base constant of the running time to $c_k' = c_k\cdot 2^{-\epsilon_k}$, where $\epsilon_k=\epsilon -  \int_{0}^{\epsilon}{R_k(r)dr}$. In particular, $\epsilon>\epsilon_k>0$. For $\epsilon\ge\frac{k-2}{k-1}$, the running time becomes $O^*(2^{o(n)})$, i.e. sub-exponential.
\end{restatable}
\begin{proof}
We follow similar arguments as in the proof of \cite{paturi2005improved} for general $k\ge 3$. For each variable in the permutation, we associate it with a uniformly random $r\in[0,1]$ which represents its ``arrival time'' according to $\pi$. Following Lemma 6,7,8 \cite{paturi2005improved},
\begin{align*}
    \prob{\pi}{x_i \notin S\text{ and is forced}}
    &\ge(1-\epsilon)\left(\int_{0}^{1}\tilde{R_k}(r)dr - \Delta_k^{(d)}\right),
\end{align*}
where $\tilde{R}_k(r)$ is the smallest nonnegative solution $\tilde{R}$ to $\left[\epsilon+(1-\epsilon)\Big(r+(1-r)\tilde{R}\Big)\right]^{k-1} = \tilde{R}$. 

Comparing to the equation for $R$, here each branch in the ``critical clause tree'' is more likely to be cut by time $r$ due to the advice: if the variable associated with the branch is in the subset advice, the branch is cut, otherwise the original recursive expression applies. Observe that by a change of variable with $u=g(r)=\epsilon + (1-\epsilon)r $, we can replace $\tilde{R_k}(r)$ with $R_k(u)$, the the smallest nonnegative solution to $R = \left[ u + (1-u)R\right]^{k-1}$, which has the same form as in the original PPSZ. 

We evaluate the probability of a variable being forced and $\expect{\pi}{G(\pi)}$,
\begin{align*}
    \prob{\pi}{x_i\notin S\text{ and is forced}}&\ge (1-\epsilon) \left(\int_{0}^{1}{\tilde{R_k}(r)dr}-\Delta_k^{(d)}\right)\\
    &\ge \int_{0}^{1}{\tilde{R_k}(r)(1-\epsilon)dr}-\Delta_k^{(d)}\\
    &= \int_{g^{-1}(\epsilon)}^{g^{-1}(1)}{R_k(g(r))g'(r)dr}-\Delta_k^{(d)}\\
    &= \int_{\epsilon}^{1}{R_k(u)du}-\Delta_k^{(d)}\\
    &=R_k -\Delta_k^{(d)} - \int_{0}^{\epsilon}{R_k(r)dr}.
\end{align*}
\begin{align}\label{eq: eps_k}
    \expect{\pi}{G(\pi)} &= n\left[(1-\epsilon) - \left(R_k -\Delta_k^{(d)} - \int_{0}^{\epsilon}{R_k(r)dr}\right) \right]\nonumber\\
    &=n\left[1-R_k +\Delta_k^{(d)} -\epsilon_k\right],
\end{align}
where $\epsilon_k=\epsilon -  \int_{0}^{\epsilon}{R_k(r)dr}$. Since $0\le R_k(r)<1$ for $r\in\left[ 0, \frac{k-2}{k-1} \right]$, for $\epsilon<\frac{k-2}{k-1}$ we have $\epsilon>\epsilon_k>0$.
Compare to $c_k =2^{1-R_k+\Delta_k^{(d)}}$ without advice, the improved base constant is $c_k'=c_k\cdot 2^{- \epsilon_k}$. 

Notably, in the original PPSZ, variables arrive ``late'' in the permutation with $r\in [\frac{k-2}{k-1},1 ]$ are forced almost surely. Given subset advice with $\epsilon>\frac{k-2}{k-1}$, if a variable is not included in the advice, the probability that it is forced goes to $1$. Under the uniqueness assumption we have $\Delta_k^{(d)}=o(1)$, so $\expect{\pi}{G(\pi)}$ becomes $o(n)$, i.e., a sub-linear number of variables are guessed, resulting in sub-exponential running time. Refer to \cite{paturi2005improved,hertli20143} for more rigorous analyses on PPSZ as well as discussions in \cite{hansen2019faster}. 
\end{proof}

\begin{restatable}[PPSZ with subset advice for \threesat]{corollary}{PPSZthreeSAT}\label{thm: subset PPSZ 3SAT}
Consider the decision problem of \threesat. Given subset advice $S$ where each variable is included independently with probability $\epsilon$, the running time of the PPSZ algorithm $O^*(c_k)^n$ can be improved exponentially, in particular, the base constant becomes $c_3 = 2^{\left(\frac{\epsilon}{1-\epsilon}+2\ln{(2-2\epsilon)}-1+o(1)\right)}$ for $\epsilon<\frac{1}{2}$ comparing to $c_3= 2^{2\ln{2} - 1 +o(1)}$ without advice, and for $\epsilon\ge\frac{1}{2}$, the running time becomes $O^*(2^{o(n)})$, i.e. sub-exponential.
\end{restatable}

\subsection{Hardness of \threesat With Subset Advice}\label{sec: subset hardness decision}
We state a hardness result for \threesat given subset advice assuming the Exponential Time Hypothesis (ETH) below. Proof is deferred to Appendix \ref{apx: ommited proofs decision}.
\begin{conjecture}[ETH \cite{impagliazzo2001complexity}]\label{ETH} There exists $\delta>0$ such that
no algorithm can solve \threesat in $O(2^{\delta n})$ time where $n$ is the number of variables.
\end{conjecture}

\begin{restatable}[Hardness of \threesat with subset advice]{theorem}{thmHardnessSat}\label{thm: subset decision hardness} Assuming the ETH, there exists $\epsilon_0>0$ such that for all $\epsilon\le\epsilon_0$, there is no polynomial time algorithm for \threesat given subset advice where each variable is included with probability $\epsilon$.
\end{restatable}

\section{Improving the Approximation Factor for the Optimization Problem}\label{sec: optimization problem}
In this section we study the optimization problem of MAX-SAT with advice in order to improve the approximation factors of polynomial time algorithms. First we show that by incorporating the subset advice in a black-box fashion into any approximation algorithm, we gain an $\Omega(\epsilon)$ improvement over the original approximation factor. Then we focus on \maxtwosat with label advice. Inspired by the work of \cite{ghoshal2024constraint,cohen2024learning}, we extend their work on MAX-$2$-LIN and adapt their techniques to the more general problem of \maxtwosat.

\subsection{MAX-SAT With Subset Advice}\label{sec: subset maxsat}
Given a subset advice $S$, we incorporate it into an approximation algorithm with the following two-step process, and state the performance guarantee in Theorem \ref{thm: subset maxsat}. 
\begin{enumerate}
    \item Set the variables in $S$ based on the advice, remove all satisfied clauses and unsatisfied literals;
    \item Run the approximation algorithm of choice on the reduced instance.
\end{enumerate} 

\begin{theorem}[MAX-SAT with subset advice]\label{thm: subset maxsat}
Consider a MAX-SAT instance and an $\alpha$-approximation algorithm. Suppose we have subset advice $S$ where each variable is included independently with probability $\epsilon$, then the approximation ratio can be improved to $\alpha + (1-\alpha)\epsilon$.\\In particular, the approximation ratio for MAX-SAT is at least $0.796+0.204\epsilon$ based on $\alpha\ge0.796$ achieved by \cite{avidor2005improved}.
\end{theorem}

\begin{proof}
Given a MAX-SAT instance with $m$ clauses on $n$ variables, suppose $m^* \le m$ clauses are satisfied in an optimal assignment $x^*$. For any clause that is satisfied in OPT, we assume (pessimistically) that the clause is satisfied by exactly one of its literals in $x^*$. Given subset advice $S$, where each variable is included independently with probability $\epsilon$, each literal's assignment in OPT is revealed with probability $\epsilon$. In particular, the satisfied literal is revealed with probability $\epsilon$, which reduces this clause. In expectation, step 1 reduces $\epsilon m^*$ of the satisfiable clauses. In step 2, $(1-\epsilon)m^*$ of the satisfiable clauses remain, and an $\alpha$-approximation algorithm will satisfy at least $\alpha(1-\epsilon)m^*$ of them.

In total, the number of satisfied clauses is at least $\epsilon m^* + \alpha(1-\epsilon)m^*$, and the approximation ratio is $\alpha' = \epsilon + \alpha(1-\epsilon)  = \alpha + (1-\alpha)\epsilon$.
\end{proof}

Our result for MAX-SAT holds in general for the optimization problem of boolean satisfiability, including \maxtwosat, \maxthreesat, and the performance ratio depends on the state-of-the-art approximation algorithm for the specific problem, as stated by the following corollaries.

\begin{corollary}[\maxtwosat with subset advice]
    Given subset advice $S$ where each variable is included independently with probability $\epsilon$, the approximation ratio for \maxtwosat is at least $0.940+0.06\epsilon$ based on $\alpha\ge0.940$ achieved by \cite{lewin2002improved}.
\end{corollary}

\begin{corollary}[\maxthreesat with subset advice]
    Given subset advice $S$ where each variable is included independently with probability $\epsilon$, the approximation ratio for \maxthreesat is at least $\frac{7}{8}+\frac{1}{8}\epsilon$ where $\alpha \ge \frac{7}{8}$ for fully satisfiable instances according to \cite{karloff19977}\footnote{It is conjectured that $\alpha\ge\frac{7}{8}$ holds for arbitrary \maxthreesat instances as well.}.
\end{corollary}

Without advice, the corresponding inapproximability bound assuming $P\ne NP$ is $0.9545$ for \maxtwosat and $\frac{7}{8}$ for \maxthreesat (even on fully satisfiable instances) \cite{haastad2001some}. Assuming the Unique Games Conjecture, the $0.940$ approximation for \maxtwosat is optimal \cite{brakensiek2024tight}.

\subsection{Hardness of \maxthreesat With Subset Advice}\label{sec: subset hardness optimization}
We state a hardness result for \maxthreesat with subset advice below, assuming the Gap Exponential Time Hypothesis (Gap-ETH). Alternatively, instead of assuming Gap-ETH, we could assume the Exponential Time Hypothesis (ETH) together with the Linear-Size PCP Conjecture (cf. \cite{dinur2016mildly,manurangsi2016birthday}).
Proof is deferred to Appendix \ref{apx: ommited proofs optimization}.
\begin{conjecture}[Gap-ETH \cite{dinur2016mildly,manurangsi2016birthday}]\label{GapETH} There exist constants $\delta,\gamma$ such that given \maxthreesat instance $\phi$, no $ O(2^{\delta n})$-time algorithm can distinguish between the case that $\mathrm{sat}(\phi)=1$ and the case that $\mathrm{sat}(\phi)\le 1-\gamma$, where $\mathrm{sat}(\phi)$ denotes the maximum fraction of satisfiable clauses. 
\end{conjecture}

\begin{restatable}[Hardness of \maxthreesat with subset advice]{theorem}{thmHardnessMaxSat}\label{thm: subset optimization hardness} Assuming the Gap-ETH, there exists $\epsilon_0=\epsilon_0(\delta,\gamma)>0$ such that for all $\epsilon\le\epsilon_0$, there is no polynomial time algorithm for \maxthreesat given subset advice with parameter $\epsilon$, such that given a satisfiable instance returns a solution that satisfies at least a $(1-\gamma)$-fraction of the clauses.
\end{restatable}

\subsection{\maxtwosat With Label Advice}\label{sec: label max2sat}
Given a label advice to an optimization problem, we first evaluate the performance of directly adopting the advice as a solution. Suppose we are given a label advice $\tilde{x}$ to a \maxeksat problem based on a ground-truth optimal assignment $x^*$. Consider a clause $C$ satisfied by $x^*$. We assume the worst-case where exactly one literal in $C$ is set to true by $x^*$ to obtain a lower-bound on the probability of $C$ being satisfied by $\tilde{x}$, which gives us the approximation factor,
\begin{align*}
    \alpha_k\ge \prob{\tilde{x}}{C \text{ is satisfied by }\tilde{x} \mid C\text{ is satisfied by }x^*} = 1-\left(\frac{1-\epsilon}{2}\right)\left(\frac{1+\epsilon}{2}\right)^{k-1}.
\end{align*}

For comparison, let $\beta_k$ denote the approximation factor of a random assignment, where $\beta_k = 1 - \frac{1}{2^k}$.
Unlike subset advice, the direct application of label advice does not immediately improve approximation performance. On the one hand, for $k = 2$, following the advice improves upon random assignment ($\alpha_2 > \beta_2$), but does not surpass the current best approximation ratio of 0.94 unless $\epsilon \ge 0.872$.
Moreover, for $k \ge 3$, unless $\epsilon$ is sufficiently large, following the label advice does not even outperform random assignment. For example, $\alpha_3 < \beta_3$ unless $\epsilon \ge 0.618$.

This motivates more refined methods for incorporating label advice, and in this section we focus exclusively on the \maxtwosat problem. We take inspiration from the prior work of \cite{ghoshal2024constraint} on the closely related MAX-CUT and MAX-$2$-LIN problems, which can be viewed as special cases of \maxtwosat via reduction. We modify the algorithm of \cite{ghoshal2024constraint} and adapt their analysis to achieve similar results for \maxtwosat.

Given a \maxtwosat formula $\phi$ with $m$ clauses and $n$ variables, we consider the $2n$ literals corresponding to the variables, i.e., pad the variables with $x_{n+1},\ldots,x_{2n}$, and replace $\bar{x_i}$ with $x_{n+i}$ in $\phi$. Following the convention of \cite{zwick2000analyzing} and \cite{lewin2002improved}, we define a vector $y\in \{-1,1\}^{2n+1}$ with respect to an assignment on literals $x \in \{0,1\}^{2n}$ in the following way: fix $y_0 = 1$ representing ``false'', and for $i=1,\ldots, 2n$: $y_i=1$ if $x_i=0$, $y_i=-1$ if $x_i =1 $.

Define the adjacency matrix on the $2n$ literals with an additional row and column of $0$ at index $0$ to match the dimension of $y$, i.e., $A \in \R^{(2n+1)\times (2n+1)}$ where 
\[ A_{ij} = 
\begin{cases}
   1 & \text{if } (x_i \lor x_j)\in \phi,\\
   0 & \text{otherwise}.
\end{cases}\]

Given $y$ and $A$ as defined above, the number of satisfied clauses equals to the integer quadratic form formulated by \cite{goemans1994879}:
\begin{align*}
    \numSAT(y) &= \sum_{(i,j)\in \phi}\frac{3-y_0y_i - y_0y_j - y_iy_j}{4} \\
    &=\frac{3}{4}m - \frac{1}{4}\sum_{(i,j)\in \phi}{(y_0y_i + y_0y_j)} - \frac{1}{4}\sum_{(i,j)\in \phi}{y_iy_j}\\
    &=\frac{3}{4}m - \frac{1}{4}\sum_{i\in[2n]}{y_0y_id_i}  - \frac{1}{8}\langle{y,Ay}\rangle \quad (\text{where }d_i \text{ is the degree of literal }i)\\
    &=\frac{3}{4}m -\frac{1}{8}f(y),
\end{align*}
where $f(y):= 2\underset{i\in[2n]}{\sum}y_0y_id_i +\langle{y,Ay}\rangle$. Note that $d_i$'s are constants given $\phi$.

In Algorithm \ref{algo: label max2sat}, we modify the objective of the quadratic program from \cite{ghoshal2024constraint} with the quadratic form for \maxtwosat and the result is stated in Theorem \ref{thm: label max2sat}. The proof follows the analysis of \cite{ghoshal2024constraint}, and we include our modified proof in Appendix \ref{apx: ommited proofs optimization} for completeness.

\begin{algorithm}[H]
\caption{\maxtwosat with label advice \cite{ghoshal2024constraint}}\label{algo: label max2sat}
    \begin{algorithmic}[1]
    \Require{
    (i) Adjacency matrix $A \in \R^{(2n+1) \times (2n+1)}$, (ii) advice vector $\tilde{y} \in \{-1, 1\}^{2n+1}$ based on the advice $\tilde{x}\in \{0,1\}^{n}$.
    }
    \Ensure{Solution $\hat{x}\in \{0,1\}^{n}$.}
    \State Solve the quadratic program:
    \begin{align*}
        \text{min } & F(y,\tilde{y}) =2\underset{i\in[2n]}{\sum}y_0y_id_i+ \langle y, A \tilde{y}/\epsilon \rangle + \onenorm{A(y - \tilde{y}/\epsilon)} \\
        \text{ subject to: }& y_0 = 1, \ y_i \in [-1,1], \ y_i = - y_{i+n}.
    \end{align*}
    \State Round the real-valued solution $y$ coordinate-by-coordinate to integer-valued solution $\hat{y}\in \{-1, 1\}^{2n+1}$ such that $f(\hat{y}) \le f(y)$.
    \State \Return $\hat{x}$ where $\hat{x_i} = -\frac{\hat{y_i}-1}{2}$ for $i=1,\ldots, n$.
    \end{algorithmic}
\end{algorithm}

\begin{restatable}[\maxtwosat with label advice]{theorem}{thmMaxTwoSatLabel}\label{thm: label max2sat}
    For an unweighted \maxtwosat instance, suppose we are given label advice $\tilde{x}$ with correct probability $\frac{1+\epsilon}{2}$ and the instance has average degree $\Delta\ge\Omega\lr{\frac{1}{\epsilon^2}}$, then Algorithm \ref{algo: label max2sat} finds solution $\hat{x}$ in polynomial time such that at least $\OPT\cdot\left(1-O\lr{\frac{1}{\epsilon \sqrt{\Delta}}}\right)$ clauses are satisfied in expectation over the randomness of the advice.
\end{restatable}

We formally define the MAX-$2$-LIN problem and show a folklore reduction to \maxtwosat.

\begin{definition}[MAX-$2$-LIN]
    In the (unweighted) MAX-$2$-LIN problem, we are given a set of Boolean variables $\{x_i\}_{i=1}^{n}$ and $m$ constraints of the form $x_i\cdot x_j = c_{ij}$ where $c_{ij} \in \{\pm1\}$. The goal is to find an assignment $\hat{x}\in \{\pm1\}^n$ that maximize the total number of satisfied constraints.
\end{definition}

Given a MAX-$2$-LIN instance on $n$ variables and $m$ constraints, we can reduce it to \maxtwosat on $n$ variables and $M=2m$ clauses in the following way:
\begin{itemize}
    \item For MAX-$2$-LIN constraint $x_i\cdot x_j = +1$, add $2$ clauses in \maxtwosat: $(x_i \lor \bar{x_j}) \land (\bar{x_i} \lor x_j)$;
    \item For MAX-$2$-LIN constraint $x_i\cdot x_j = -1$, add $2$ clauses in \maxtwosat: $(x_i \lor x_j) \land (\bar{x_i} \lor \bar{x_j})$.
\end{itemize}
This reduction preserves approximation as stated in Proposition \ref{prop: max2lin to max2sat} below, and given label advice, our Theorem \ref{thm: label max2sat} is a generalization of Theorem 1.4 from \cite{ghoshal2024constraint}. Proof is deferred to Appendix \ref{apx: ommited proofs optimization}. 

\begin{restatable}[]{proposition}{propLinToSat}\label{prop: max2lin to max2sat}
Under the reduction above, a $\left(1-O(\delta)\right)$-approximation to \maxtwosat corresponds to a $\left(1-O(\delta)\right)$-approximation to MAX-$2$-LIN.
\end{restatable} 

\section{Discussion} \idan{sketch}
In this paper, we showed how subset advice can be incorporated to improve the running time of algorithms for $k$-SAT that currently achieve the best known performance in worst-case analysis. 
For the optimization problem MAX-SAT and its variants, we incorporated subset advice into any algorithm and showed that the approximation factor improves linearly with the advice parameter. We also proved that, assuming ETH and Gap-ETH, these are the best possible results for \threesat and \maxthreesat. Using label advice, we obtained near-optimal results for \maxtwosat instances where the average degree exceeds a threshold depending only on the amount of advice. This generalizes previous results for MAX-CUT and MAX-$2$-LIN. Open questions regarding the label advice include designing algorithms for \ksat, as well as  incorporating label advice into SDP-based methods to solve more general problems of \maxksat. An interesting direction for future work is to explore and compare different advice models, such as proving a formal separation between the label and subset advice models as it is plausible that the label advice model constitutes a weaker model. We would also like to explore a variation on the label advice model, where we are allowed to make a few queries to an oracle with non-persistent noise. 

\section*{Acknowledgments}

We thank Suprovat Ghoshal for helpful discussions. 
This research was supported in part by the National Science Foundation grant ECCS-2217023.

\newpage
\bibliographystyle{alpha}
\bibliography{refs}

\newpage
\appendix

\section{Omitted Proofs From Section \ref{sec: decision problem}}\label{apx: ommited proofs decision}
\PPSZthreeSAT*
\begin{proof}
Recall that $R_k(r)$ is the smallest nonnegative solution $R$ to $\Big(r+(1-r)R\Big)^{k-1} = R$.
Now we focus on the case that $k=3$ and solve for $R_3(r)$, 
\begin{align*}
    R_3(r)&=\begin{cases}
        \left(\frac{r}{1-r}\right)^2 & 0\le r\le \frac{1}{2}\\
        1 & \frac{1}{2}\le r\le1
    \end{cases}.
\end{align*}
Then we evaluate $\epsilon_3$,
\begin{align*}
    \epsilon_3&=\epsilon -  \int_{0}^{\epsilon}{R_3(r)dr}\\
    &=\epsilon - \left(\frac{1}{1-\epsilon}+2\ln{(1-\epsilon)}-1+\epsilon \right)\\
    &=-\frac{1}{1-\epsilon}-2\ln{(1-\epsilon)}+1.
\end{align*}

Now we apply Theorem \ref{thm: subset PPSZ kSAT}. Recall that without advice $R_3 = 2- 2\ln{2}$ and $c_3 =2^{ 2\ln{2} - 1 +o(1)}$, therefore
\begin{align*}
    c_3' &= c_3 \cdot 2^{- \epsilon_3}\\
         &= 2^{2\ln{2} - 1 +o(1)- \epsilon_3}\\
         &=2^{\left(\frac{\epsilon}{1-\epsilon}+2\ln{(2-2\epsilon)}-1+o(1)\right)}.
\end{align*}
\end{proof}

\thmHardnessSat*
\begin{proof}
Let $\epsilon_0<\delta$, where $\delta$ is the constant in the ETH. Suppose there is a polynomial time algorithm that solves \threesat given subset advice with parameter $\epsilon\le\epsilon_0$. Fix a subset of size $\epsilon n$, we can simulate a subset advice by enumerating all possible assignments then run the algorithm, thereby solving \threesat in time $O(2^{\epsilon n}\mathrm{poly}(n))\le O(2^{\delta n})$, contradicting ETH.
\end{proof}

\section{Omitted Proofs From Section \ref{sec: optimization problem}}\label{apx: ommited proofs optimization}
\thmHardnessMaxSat*
\begin{proof}
Let $\epsilon_0<\delta$ in the ETH. Suppose there is a polynomial time algorithm that given a fully satisfiable instance of \maxthreesat and subset advice with parameter $\epsilon\le\epsilon_0$, returns a solution satisfying at least a $(1-\gamma)$-fraction of the clauses. Given input $\phi$, we fix a subset of size $\epsilon n$ and simulate a subset advice by enumerating all possible assignments, then run the algorithm. If $\phi$ is fully satisfiable, eventually we will get a solution satisfying at least a $(1-\gamma)$-fraction of the clauses; on the other hand, if $\mathrm{sat}(\phi)<1-\gamma$, no solution can satisfy a $(1-\gamma)$-fraction, thereby we can distinguish between the two cases in time $O(2^{\epsilon n}\mathrm{poly}(n))\le O(2^{\delta n})$, contradicting Gap-ETH.
\end{proof}

\thmMaxTwoSatLabel*
\begin{proof} 
The proof follows from a chain of inequalities in expectation over the randomness the advice.
\begin{enumerate}
    \item $f(\hat{y}) \le f(y)$, where $y$ is the QP solution, and $\hat{y}$ is the rounding of $y$. This follows from the same argument as in \cite{ghoshal2024constraint}.
    \item $f(y)\le F(y,\tilde{y})$, where $F(y,\tilde{y})$ is the minimum value of the QP attained at solution $y$. This step follows from Lemma \ref{lemma: f(y) to F(y)}, which is an extension of Claim 3.2 from \cite{ghoshal2024constraint}.
    \item $F(y,\tilde{y}) \le F(y^*,\tilde{y})$, where $y^*$ is the vector corresponding to the ground-truth optimal solution $x^*$. This step directly follows from the optimality of the QP solution $y$, and the fact that $y^*$ is feasible to the QP.
    \item $ F(y^*,\tilde{y}) \le f(y^*) + m\cdot O(\frac{1}{\epsilon \sqrt{\Delta}})$. This step follows from Lemma \ref{lemma: F(y*) to f(y*)}, which is an extension of Lemma 3.3 from \cite{ghoshal2024constraint}.
\end{enumerate}
Putting these together, and we may assume $\OPT \ge \frac{3}{4}m$,
\begin{align*}
    f(\hat{y})&\le f(y^*) + m\cdot O\left(\frac{1}{\epsilon \sqrt{\Delta}}\right), \\
    \numSAT(\hat{y})&= \frac{3}{4}m - \frac{1}{8}f(\hat{y}) \\
    &\ge \frac{3}{4}m - \frac{1}{8}f(y^*) -  m\cdot O\left(\frac{1}{\epsilon \sqrt{\Delta}}\right)\\
    &\ge \OPT - \OPT\cdot  O\left(\frac{1}{\epsilon \sqrt{\Delta}}\right)\\
    &\ge \OPT \cdot \left(1-O\left(\frac{1}{\epsilon \sqrt{\Delta}}\right)\right) .
\end{align*}
\end{proof}

\begin{restatable}[]{lemma}{lemmafToF}\label{lemma: f(y) to F(y)}
    For $y \in [-1,1]^{2n+1}$, $f(y)\le F(y,\tilde{y})$.
\end{restatable}
\begin{proof}This Lemma is a modified version of Claim 3.2 \cite{ghoshal2024constraint}, we include a full proof for completeness.
    \begin{align*}
        f(y) &= 2\underset{i\in[2n]}{\sum}y_0y_id_i +\langle{y,Ay}\rangle \\
        &= 2\underset{i\in[2n]}{\sum}y_0y_id_i +\langle{y,A\tilde{y}/\epsilon}\rangle + \langle{y,A(y-\tilde{y}/\epsilon)}\rangle \\
        &\le 2\underset{i\in[2n]}{\sum}y_0y_id_i +\langle{y,A\tilde{y}/\epsilon}\rangle + \inftynorm{y}\cdot \onenorm{A (y - \tilde{y}/\epsilon)}\\
        & = 2\underset{i\in[2n]}{\sum}y_0y_id_i +\langle{y,A\tilde{y}/\epsilon}\rangle +\onenorm{A (y - \tilde{y}/\epsilon)}\\
        & = F(y,\tilde{y}),
    \end{align*}
where the inequality follows from Hölder's inequality.
\end{proof}

\begin{restatable}[]{lemma}{lemmaFtof}\label{lemma: F(y*) to f(y*)}
    $\expect{\tilde{y}}{F(y^*,\tilde{y})} \le f(y^*) + \frac{2}{\epsilon}\sqrt{mn}$. 
    \\Note that $\frac{2}{\epsilon}\sqrt{mn} = m\cdot O(\frac{1}{\epsilon \sqrt{\Delta}})$, following the definition of $\Delta = \frac{m}{n}$.
\end{restatable}
\begin{proof}   
Let $z =  y^* - \tilde{y}/\epsilon$. First we calculate the mean and variance of $\tilde{y_i}$ and $z_i$. 
\begin{align*}
    &\expect{}{\tilde{y_i}} = \frac{1+\epsilon}{2} y_i^* + \frac{1-\epsilon}{2}(-y_i^*) =  \epsilon y_i^*,\\
    &\expect{}{z_i} = 0,\\
    &\expect{}{z_i^2} = \frac{1+\epsilon}{2}(y_i^* - y_i^*/\epsilon )^2 + \frac{1-\epsilon}{2}( y_i^* + y_i^*/\epsilon)^2 =\frac{1-\epsilon^2}{\epsilon^2}.
\end{align*}
Consider
\begin{align*}
    \expect{\tilde{y}}{F(y^*,\tilde{y})} &= 2\underset{i\in[2n]}{\sum}y^*_0y^*_id_i +\expect{\tilde{y}}{\langle y^*,A \tilde{y}/\epsilon\rangle} +  \expect{\tilde{y}}{\onenorm{Az}}\\
    &=2\underset{i\in[2n]}{\sum}y^*_0y^*_id_i +\langle{y^*,A\expect{\tilde{y}}{\tilde{y}/\epsilon}}\rangle+  \expect{\tilde{y}}{\onenorm{Az}}\\
    &=2\underset{i\in[2n]}{\sum}y^*_0y_i^*d_i +\langle{y^*,Ay^*}\rangle+  \expect{\tilde{y}}{\onenorm{Az}}\\
    &=f(y^*) +  \expect{\tilde{y}}{\onenorm{Az}}
\end{align*}

According to the proof of Lemma 3.3 \cite{ghoshal2024constraint}, the expectation term
\begin{align*}
    \expect{\tilde{y}}{\onenorm{Az}} 
    &\le \sqrt{2n} \ \expect{\tilde{y}}{\twonorm{Az}}\\
    &\le \sqrt{2n} \sqrt{\expect{}{z_i^2}\cdot \fnorm{A}^2 } \\
    &\le \sqrt{2n}\sqrt{\frac{1-\epsilon^2}{\epsilon^2}}\fnorm{A}\\
    &\le \frac{2}{\epsilon}\sqrt{mn},
\end{align*}
since $\fnorm{A}^2 = 2m$.
\end{proof}

\propLinToSat*
\begin{proof}
Recall the reduction:
given a MAX-$2$-LIN instance on $n$ variables and $m$ constraints, we can reduce it to \maxtwosat on $n$ variables and $M=2m$ clauses in the following way:
\begin{itemize}
    \item For MAX-$2$-LIN constraint $x_i\cdot x_j = +1$, add $2$ clauses in \maxtwosat: $(x_i \lor \bar{x_j}) \land (\bar{x_i} \lor x_j)$;
    \item For MAX-$2$-LIN constraint $x_i\cdot x_j = -1$, add $2$ clauses in \maxtwosat: $(x_i \lor x_j) \land (\bar{x_i} \lor \bar{x_j})$.
\end{itemize}
Given a solution $x$ to the \maxtwosat problem, we can translate it to a solution $z$ to the MAX-$2$-LIN problem by setting $z_i = 1$ if $x_i=1$, and $z_i= -1$ if $x_i=0$.

Notice that for each of the constraint and its $2$ corresponding clauses, the constraint is satisfied if and only if both clauses are true, and the constraint is not satisfied if and only if exactly one of the clauses is true. Therefore, $\numLIN = \numSAT - m$. 
Denote the value of an optimal solution to MAX-$2$-LIN as $m^*$, and the value of an optimal solution to \maxtwosat as $M^*$, and notice that $m^* = M^* - m$.
Given an assignment $x$ that satisfies $\left(1-O(\delta)\right)\cdot M^*$ clauses in \maxtwosat, the number of constraints in MAX-$2$-LIN satisfied by corresponding $y$ is
\begin{align*}
    \numLIN &= \numSAT - m\\
    &=\left(1-O(\delta)\right)\cdot M^* - m\\
    &=\left(1-O(\delta)\right)\cdot(m^*+m) - m\\
    &=\left(1-O(\delta)\right) \cdot m^* - O(\delta)\cdot m\\
    &=\left(1-O(\delta)\right) \cdot m^*, \text{ since we may assume } m^* = \Theta(m).
\end{align*}

Given label advice, our Theorem \ref{thm: label max2sat} is a generalization of Theorem 1.4 from \cite{ghoshal2024constraint}. By mapping $\{-1,1\}$ to $\{0,1 \}$, an advice to the MAX-$2$-LIN instance can be translated as an advice to the \maxtwosat instance with the same $\epsilon$. Furthermore, the degree of the MAX-$2$-LIN instance is $\Delta = \frac{2m}{n}$ which is equal to the degree of the \maxtwosat instance $\frac{M}{n}$, so the same average degree assumption applies to both problems as well. Therefore we recover the solution to MAX-$2$-LIN and generalize the previous results from the ``symmetric" constraint satisfaction problems to ``non-symmetric" SAT problems.
\end{proof}

\end{document}